\documentclass[11pt]{article}
\usepackage[margin=1in]{geometry}
\usepackage{palatino}
\usepackage{amsthm,amsmath,amssymb,amsfonts}
\usepackage[final]{graphicx}
\usepackage{subcaption}
\usepackage{hyperref}
\usepackage{booktabs}
\usepackage{algorithm, algpseudocode}
\usepackage{authblk}
\usepackage{xcolor}
\usepackage{ifdraft}
\usepackage[utf8]{inputenc}
\usepackage{xspace}
\usepackage[final]{pdfpages}
\usepackage{graphicx}
\usepackage{verbatim}
\usepackage{xfrac}
\usepackage{array}
\usepackage[normalem]{ulem}

\newcommand{\dhcre}{Detailed DHC-B\xspace}

\newcommand{\safetabh}{SafeTab-H\xspace}
\usepackage{fancyhdr}
\fancypagestyle{empty}{
   \fancyhf{} 
   
}

\theoremstyle{definition}
\newtheorem{definition}{Definition}

\theoremstyle{plain}
\newtheorem{theorem}{Theorem}
\newtheorem{lemma}{Lemma}
\newtheorem{corollary}{Corollary}
\newtheorem{proposition}{Proposition}

\newcommand{\eat}[1]{}

\graphicspath{{screenshots/}{figures/}}

\title{SafeTab-H: Disclosure Avoidance for the 2020 Census Detailed Demographic and Housing Characteristics File B (Detailed DHC-B)}
\author[1]{William Sexton}
\author[1]{Skye Berghel}
\author[1]{Bayard Carlson}
\author[1]{Sam Haney}
\author[1]{Luke Hartman}
\author[1]{Michael Hay}
\author[1]{Ashwin Machanavajjhala}
\author[1]{Gerome Miklau}
\author[1]{Amritha Pai}
\author[1]{Simran Rajpal}
\author[1]{David Pujol}
\author[1]{Ruchit Shrestha}
\author[1]{Daniel Simmons-Marengo}

\affil[1]{Tumult Labs}
\date{May 28, 2024}

\begin{document}

\maketitle
\begin{abstract}
This article describes SafeTab-H, a disclosure avoidance algorithm applied to the release of the U.S. Census Bureau's Detailed Demographic and Housing Characteristics File B (Detailed DHC-B) as part of the 2020 Census. The tabulations contain household statistics about household type and tenure iterated by the householder's detailed race, ethnicity, or American Indian and Alaska Native tribe and village at varying levels of geography. We describe the algorithmic strategy which is based on adding noise from a discrete Gaussian distribution and show that the algorithm satisfies a well-studied variant of differential privacy, called zero-concentrated differential privacy. We discuss how the implementation of the SafeTab-H codebase relies on the Tumult Analytics privacy library. We also describe the theoretical expected error properties of the algorithm and explore various aspects of its parameter tuning.
\end{abstract}

\newpage
\tableofcontents
\newpage

\section{Introduction}

It is the responsibility of the U.S. Census Bureau (Census Bureau) to conduct a census of the U.S. population every 10 years. The 2020 Census is the latest of such efforts, which aims to enumerate every person living in the United States. As part of the 2020 Census undertaking, the Census Bureau manages 35 operations (e.g., Address Canvassing, Nonresponse Followup, and Redistricting Data Program) \cite{census-operations}. Each of these operations controls a number of systems that handle various aspects of the entire census endeavor, ranging from data collection and processing to the dissemination of data products to the U.S. people. Throughout these 2020 Census procedures, the Census Bureau strives to maintain the privacy and confidentiality of its respondents. The Disclosure Avoidance System (DAS) manages the confidentiality protection of statistical data releases from the 2020 Census. The DAS executes its duties after census responses have been collected and processed into a database known as the Census Edited File (CEF) but before data products are released for public consumption. For the 2020 Census, the DAS was redesigned to modernize its privacy protection mechanisms. The modernization effort provides individuals and households with state-of-the-art protection against the privacy threats associated with releasing census data to the public.

The Census Bureau releases several different data products for the 2020 Census, including the 2020 Census Redistricting Data (P.L. 94-171) Summary File, Demographic and Housing Characteristics File (DHC), Demographic Profile, Detailed Demographic and Housing Characteristics File A (Detailed DHC-A), Detailed Demographic and Housing Characteristics File B (Detailed DHC-B), and Supplemental Demographic and Housing Characteristics File (S-DHC). Each data product exhibits distinct challenges concerning confidentiality protection. Hence, the DAS deploys various algorithms to optimize protection and accuracy across each data product. We note that some data products share similar enough challenges to utilize the same algorithm. Despite algorithmic differences, all statistical disclosure limitation techniques fit into the same overarching privacy framework known as \emph{differential privacy}. The differential privacy framework calls for the design of algorithms that satisfy mathematically provable guarantees regarding the data publication process. Section \ref{sec:privacy-defintions} gives further details. Although the differential privacy framework encompasses several different privacy definitions, in this paper, one should assume we use the term "differential privacy" to refer to zero-concentrated differential privacy unless otherwise specified.

SafeTab-P and SafeTab-H are two of the privacy algorithms deployed by the DAS. SafeTab-P provides privacy protection for the Detailed DHC-A whereas SafeTab-H was designed to provide differential privacy guarantees for the Detailed DHC-B. SafeTab-H is the primary topic of this paper, but we mention SafeTab-P for two reasons: (1) both algorithms share many similarities because the Detailed DHC-A and Detailed DHC-B are closely related data products and (2) some outputs of the SafeTab-P algorithm appear as inputs to the SafeTab-H algorithm. We refer the reader to Section \ref{sec:problem} for an in-depth description of the data product and privacy release problem.

The main goals of this article are threefold:

\begin{enumerate}
    \item Describe the SafeTab-H algorithm and how it meets the requirements of the Detailed DHC-B.
    \item Prove the privacy properties of the SafeTab-H algorithm.
    \item Describe the parameters in the SafeTab-H algorithm and how they impact privacy-accuracy trade-offs.
\end{enumerate}

For the first point, we provide a technical pseudocode description of SafeTab-H that details how privacy protection is applied to create the Detailed DHC-B tabular summaries (Section \ref{sec:algorithm-description}). We also highlight salient differences between our pseudocode abstraction and the programmed codebase of the algorithm (Section \ref{sec:implementation}). While this article provides meaningful context regarding the implemented algorithm, it does not provide a detailed overview of the code architecture (e.g., module interactions and class descriptions). The SafeTab-H codebase is open source for interested readers.

For the second point, we provide relevant background material on the differential privacy framework and explain why SafeTab-H adheres to the framework (Section \ref{sec:safetab-discrete-gaussian-privacy}).

For the third point, we discuss the parameters in SafeTab-H that impact the privacy-accuracy trade-offs of the algorithm. We cover parameter tuning, including a brief look at related data accuracy considerations (Section \ref{sec:params}).

\section{Problem Setup}
\label{sec:problem}
The Detailed DHC-B contains statistics (counts) of household type and tenure, including total household count, for households in the United States and Puerto Rico, crossed with detailed races and ethnicities at varying levels of geography. The Detailed DHC-B includes data for 300 detailed race and ethnicity groups and 1,187 American Indian and Alaska Native tribes and villages. Geographies include nation, state, county, census tract, place, and American Indian/Alaska Native/Native Hawaiian (AIANNH) areas. The Detailed DHC-A included total population and sex by age data for the same detailed races and ethnicities and geographies. The Census Bureau published the Detailed DHC-A in September 2023. In this section, we define relevant concepts, outline the statistics to be released, and then formulate the differentially private algorithm design problem. 

\subsection{Households}
A household is an individual or group of individuals living together in an occupied housing unit. The Detailed DHC-B does not include vacant housing units. For our purposes, we treat ``household'' and ``occupied housing unit'' as interchangeable terms. In the CEF, every household has exactly one person designated as the householder. Every household has a household type (e.g., married couple family household) and a tenure status (e.g., renter occupied).   

\subsection{Geography}
\label{sec:geography}
Every household is located in exactly one Census block that determines its geographic location. Census blocks are the most granular form of \textit{geographic entities}. All other geographic entities (e.g., Los Angeles County, the state of California, and the United States) are aggregations of Census blocks. Geographic entities are divided into \emph{geographic summary levels}. A geographic summary level is a set of nonoverlapping geographic entities, such as the set of all states or the set of all counties.  The Detailed DHC-B produces statistics for the following geographic summary levels:

\begin{itemize}
    \item Nation
    \item State or State equivalents
    \item County or County equivalents
    \item Census Tract
    \item Place
    \item AIANNH areas.
\end{itemize}

Henceforth, we tend to write State (County) without including the ``or State (County) equivalent'' qualifier, although one should assume the qualifier when applicable. Washington, D.C. is an example of a state equivalent. With some context-specific exceptions, this document adopts the convention of capitalizing references to levels and using lowercase when referencing an entity or entities within a level (e.g., there are over 3000 counties in the County level, and there are four counties in the state of Rhode Island).

\subsection{Race and Ethnicity}
\label{sec:race-defs}
In the 2020 Census, the Census Bureau collected detailed race and ethnicity data from individuals in accordance with the 1997 Federal Register Notice ``Revisions to the Standards for the Classification of Federal Data on Race and Ethnicity'' released by the Office of Management and Budget (OMB). 

Per these guidelines, every household is associated with one or more \textit{race codes} of its householder and a single \textit{ethnicity code} of its householder. That is, a household's race and ethnicity assignment is based solely on the attributes of the householder, even if other individuals in the household have differing race or ethnicity codes. The maximum number of race codes, called the \textit{race multiplicity}, that a householder, and hence household, can be associated with is limited to eight by the 2020 Census data collection procedures. 

A \textit{race group} is a set of race codes (e.g., German is defined by race codes ranging from 1170 through 1179). Similarly, an \textit{ethnicity group} is a set of ethnicity codes (e.g., Mexican is defined by ethnicity codes ranging from 2010 through 2099). Hence, householders with differing race or ethnicity codes can nonetheless belong to the same race or ethnicity groups.

\textit{Detailed} race or ethnicity groups are the most disaggregated racial or ethnic group classifications for which the Census Bureau publishes data.  Examples of detailed racial or ethnic groups include Lebanese, Dutch, Guatemalan, Puerto Rican, Ethiopian, Nigerian, Mongolian, Thai, Brazilian, Belizean, Samoan, Marshallese, Chevak Native Village, and Navajo Nation. The \textit{major} racial categories are aggregated race groupings that represent the minimum allowable racial categories for which census data may be published. The major racial categories in the 2020 Census were White, Black or African American, American Indian or Alaska Native, Asian, Native Hawaiian or Other Pacific Islander, and Some Other Race. The aggregated ethnic equivalent of the major racial categories is a coarse binary classification (Hispanic or Latino, Not Hispanic or Latino). \textit{Regional} race or ethnicity groups provide an intermediate level of aggregation. Examples of regional racial or ethnic groups include European, Central American, Caribbean, Sub-Saharan African, Alaska Native, and American Indian. Subject-matter experts at the Census Bureau determined which race or ethnicity groups are detailed and which groups are regional. For the purposes of SafeTab-H, we take these classifications as given exogenous factors. The universe specification of valid detailed race and ethnicity groups and their classification into detailed or regional groups occurred before data collection for the 2020 Census \cite{race-blog}. Appendix G of \cite{race-spec} provides a complete enumeration of the detailed race and ethnicity groups. 

A household is in a race group \textit{alone} if all race codes associated with its householder are contained in the race group. For example, if a householder self-identifies with the single race code for Navajo Nation and no other race codes, the household would belong to the detailed race group Navajo Nation alone. Alternatively, a householder may report multiple race groups (e.g., British, Scottish, and Dutch) that aggregate into the same regional group (European alone). A household is in a race group \textit{alone or in any combination} if any race code associated with its householder is contained in the race group. This concept pertains to households where the householder self-identifies with a single detailed race (e.g., British) or with multiple detailed races (e.g., British and Thai). In both examples, the household belongs to the detailed British race group alone or in any combination. The household also belongs to the regional European race group alone or in any combination. Since all householders are only associated with a single ethnicity code, respondents may only be in one detailed ethnicity group and in one regional ethnicity group.

A \textit{race characteristic iteration} is a race group combined with the specification of either ``alone'' or ``alone or in any combination'' (e.g., Latin American Indian alone or in any combination is a characteristic iteration). An \textit{ethnicity characteristic iteration} is synonymous with an ethnicity group. Ethnicity characteristic iterations do not carry either ``alone'' or ``alone or in any combination'' designators. One household may be associated with multiple characteristic iterations. Like geographical entities, characteristic iterations are also divided into \emph{characteristic iteration levels}. We have already provided the defining aspect of these iteration levels: namely, the concepts of \emph{detailed} and \emph{regional} race groups. We adopt the convention of capitalizing references to Detailed and Regional levels while using lowercase for references to detailed and regional iterations within a level. The Detailed characteristic iteration level consists of the set of characteristic iterations for all detailed race groups either alone or alone or in any combination (e.g., Japanese alone, Japanese alone or in any combination, Celtic alone, and Celtic alone or in any combination) and all detailed ethnicity groups. The Regional characteristic iteration level consists of the set of characteristic iterations for all regional race alone or alone or in any combination (e.g., Middle Eastern or North African alone, Middle Eastern or North African alone or in any combination, Polynesian alone, and Polynesian alone or in any combination) as well as all regional ethnicity groups. We intentionally omit the notion of a major race and ethnicity characteristic iteration level, as no statistics for this level are produced by the SafeTab-H algorithm for the Detailed DHC-B.

\subsection{Population Groups}\label{sec:population-group-levels}

A \textit{population group} is a pair $(g, c)$, where $g$ is a geographic entity (e.g., the state of North Carolina), and $c$ is a characteristic iteration (e.g., Latin American Indian alone or in any combination). Population groups are divided into \textit{population group levels}. We will often identify a population group level by specifying a (geography level, characteristic iteration level) pair. However, each population group level is really a set of population groups, where each population group's geographic entity belongs to the specified geography level and its characteristic iteration belongs to the specified characteristic iteration level. More formally, the Detailed DHC-B requires the publication of statistics for the following population group levels:

\begin{itemize}
    \item (Nation, Detailed) $\equiv$ $\{(g, c): g \text{ is the nation}, c \text{ is a detailed characteristic iteration}\}$
    \item (State, Detailed) $\equiv$ $\{(g, c): g \text{ is a state}, c \text{ is a detailed characteristic iteration}\}$
    \item (County, Detailed) $\equiv$ $\{(g, c): g \text{ is a county)}, c \text{ is a detailed characteristic iteration}\}$
    \item (Census Tract, Detailed) $\equiv$ $\{(g, c): g \text{ is a Census tract}, c \text{ is a detailed characteristic iteration}\}$
    \item (Place, Detailed) $\equiv$ $\{(g, c): g \text{ is a place}, c \text{ is a detailed characteristic iteration}\}$
    \item (AIANNH, Detailed) $\equiv$ $\{(g, c): g \text{ is an AIANNH area}, c \text{ is a detailed characteristic iteration}\}$
    \item (Nation, Regional) $\equiv$ $\{(g, c): g \text{ is the nation}, c \text{ is a regional characteristic iteration}\}$
    \item (State, Regional) $\equiv$ $\{(g, c): g \text{ is a state}, c \text{ is a regional characteristic iteration}\}$
    \item (County, Regional) $\equiv$ $\{(g, c): g \text{ is a county}, c \text{ is a regional characteristic iteration}\}$
    \item (Census Tract, Regional) $\equiv$ $\{(g, c): g \text{ is a Census tract}, c \text{ is a regional characteristic iteration}\}$
    \item (Place, Regional) $\equiv$ $\{(g, c): g \text{ is a place}, c \text{ is a regional characteristic iteration}\}$
\end{itemize}

In practice, some detailed or regional characteristic iterations may be omitted from the tabulations in a geography level. In other words, the above population group levels are supersets of the actual population group levels. This is done in accordance with specifications for the Detailed DHC-B provided by the Census Bureau. There is no concise representation for the exact level sets. The population group level (AIANNH, Regional) is intentionally omitted from the Detailed DHC-B.

One household may belong to multiple population groups in the set that comprises a population group level. For example, a household located in Texas with a householder reporting Kenyan and Ghanaian detailed races would belong to the (Texas, Kenyan alone or in any combination) and the (Texas, Ghanaian alone or in any combination) population groups which are both contained in the population group level identified by (State, Detailed). A household located in Schuyler County, NY with a householder reporting a single detailed race of Dutch would still belong to the (Schuyler County, NY, Dutch alone) and the (Schuyler County, NY, Dutch alone or in any combination) population groups which are both contained in the level identified by (County, Detailed). One household may be connected with population groups across multiple population group levels. The household with the Dutch householder residing in Schuyler County, NY would additionally belong to the (NY, Dutch alone) and (NY, Dutch alone or in any combination) population groups in the (State, Detailed) level, the (Schuyler County, NY, European alone) and (Schuyler County, NY, European alone or in any combination) population groups in the (County, Regional) level, etc. It is possible for a household to not belong to any population groups in a particular level. Specifically, because AIANNH areas do not cover the United States, a household located outside all designated AIANNH areas does not belong to any population groups in the (AIANNH, Detailed) level. For example, households in Arkansas do not belong to any AIANNH areas, and therefore would not contribute to any AIANNH area counts.  

A householder's characteristic iterations primarily determine the number of population groups the household belongs to in each level because the geographic entities in a geography level are disjoint (a household cannot be located in both Schuyler County, NY and Fairfax County, VA). A household associated with the maximum of eight race codes and a single ethnicity code could belong to at most nine detailed characteristic iterations (eight alone or in any combination race groups and one ethnic group) and, similarly, at most nine regional characteristic iterations. Thus, for any given population group level, the maximum number of population groups a household may contribute to is nine.

\subsection{Detailed Demographic and Housing Characteristics File B}
The Detailed DHC-B aims to tabulate statistics about household type and tenure by population groups. The following statistical tables are released for each eligible population group.\footnote{Population groups may be deemed ineligible to receive certain statistics for various reasons discussed throughout the document. For instance, Detailed DHC-B tables are not released for groups that do not appear in the Detailed DHC-A's T01001.}

\begin{itemize}
    \item Household type counts for eligible population groups. The household type tables come in four different variants (T03001, T03002, T03003, and T03004) as shown in Tables \ref{tab:t03001-shell}, \ref{tab:t03002-shell}, \ref{tab:t03003-shell}, and \ref{tab:t03004-shell}. For convenience in mathematical expressions, we use HT to refer generically to the household type table class.  
    \item Tenure counts for eligible population groups. The tenure tables come in two different variants (T04001 and T04002) as shown in Tables \ref{tab:t04001-shell} and \ref{tab:t04002-shell}. For convenience in mathematical expressions, we use the abbreviation T to refer generically to the tenure table class.  
\end{itemize}

Each table has a \emph{basis}, a set of fine-grained, disaggregated table cells from which the remaining cells may be generated. In the below table shells, the dark text indicates which cells form the table's basis, while the light text shows the aggregated table cells. For a given table, every household can be assigned to exactly one category of the table's basis. For example, with T03002, every household is exclusively either a family household or a nonfamily household. 

\begin{table}[H]
\begin{tabular}{l}
\textbf{T03001. Household Type (Universe) }\\
\textit{Universe: Households}\\
Total\\
\end{tabular}
\caption{\label{tab:t03001-shell} T03001 consists solely of a total household count.}
\end{table}

\begin{table}[H]
\begin{tabular}{l}
\textbf{T03002. Household Type (2 Categories)}\\
\textit{Universe: Households}\\
\textcolor{lightgray}{Total:}\\
\quad \quad Family Household\\
\quad \quad Nonfamily Household\\
\end{tabular}
\caption{\label{tab:t03002-shell} T03002 consists of counts for family household and nonfamily household.}
\end{table}

\begin{table}[H]
\begin{tabular}{l}
\textbf{T03003. Household Type (6 Categories)}\\
\textit{Universe: Households}\\
\textcolor{lightgray}{Total:}\\
\quad \quad \textcolor{lightgray}{Family Household:}\\
\quad \quad \quad \quad Married Couple Family\\
\quad \quad \quad \quad Other Family\\
\quad \quad \textcolor{lightgray}{Nonfamily Household:}\\
\quad \quad \quad \quad Householder Living Alone\\
\quad \quad \quad \quad Householder Not Living Alone\\
\end{tabular}
\caption{\label{tab:t03003-shell} T03003 contains counts for married couple households, other family households, and the breakdown of nonfamily households into categories for householders living alone or not living alone.}
\end{table}

\begin{table}[H]
\begin{tabular}{l}
\textbf{T03004. Household Type (8 Categories)}\\
\textit{Universe: Households}\\
\textcolor{lightgray}{Total:}\\
\quad \quad \textcolor{lightgray}{Family Household:}\\
\quad \quad \quad \quad Married Couple Family\\
\quad \quad \quad \quad \textcolor{lightgray}{Other Family:}\\
\quad \quad \quad \quad \quad \quad Male householder, no spouse present\\
\quad \quad \quad \quad \quad \quad Female householder, no spouse present\\
\quad \quad \textcolor{lightgray}{Nonfamily Household:}\\
\quad \quad \quad \quad Householder Living Alone\\
\quad \quad \quad \quad Householder Not Living Alone\\
\end{tabular}
\caption{\label{tab:t03004-shell} T03004 contains counts for married couple households, the breakdown of other family households into categories for male or female householders with no spouse present, and the breakdown of nonfamily households into categories for householders living alone or not living alone.}
\end{table}

\begin{table}[H]
\begin{tabular}{l}
\textbf{T04001. Tenure (Universe)}\\
\textit{Universe: Occupied Housing Units}\\
Total\\
\end{tabular}
\caption{\label{tab:t04001-shell} T04001 consists solely of a total household count.}
\end{table}

\begin{table}[H]
\begin{tabular}{l}
\textbf{T04002. Tenure (3 Categories)}\\
\textit{Universe: Occupied Housing Units}\\
\textcolor{lightgray}{Total:}\\
\quad \quad Owned with a mortgage or a loan\\
\quad \quad Owned free and clear\\
\quad \quad Renter Occupied\\
\end{tabular}
\caption{\label{tab:t04002-shell} T04002 consists of counts for three household tenure categories: owned with a mortgage or a loan, owned free and clear, and renter occupied.}
\end{table}

As previously mentioned, despite the difference in universe terminology (i.e., household vs occupied housing units), the household type tables and tenure tables both provide total counts over the same entities. That is, a population group's total household count (e.g., T03001) measures the same quantity as the population group's total occupied housing unit count (e.g., T04001).

An additional table pertaining to properties of persons is also needed as an input for the SafeTab-H algorithm. The T01001 table is produced as part of the Detailed DHC-A.

\begin{table}[H]
\begin{tabular}{l}
\textbf{T01001 Total Population}\\
\textit{Universe: Total Population}\\
Total\\
\end{tabular}
\caption{\label{tab:t1-shell} The T01001 table contains counts of persons, instead of households, and is iterated by detailed race and ethnicity of each individual in the population rather than of the householder. This table is an output of SafeTab-P and was released in the Detailed DHC-A.}
\end{table}

\subsection{Private Release Problem}
The Census Bureau release of statistical data products is regulated under Title 13, which disallows any data publications in which an individual's data can be identified \cite{title13}. Moreover, it has been demonstrated that legacy statistical disclosure limitation techniques are vulnerable to attacks that can reconstruct the sensitive person records from aggregate statistics \cite{ap-census-attack}. Hence, the Census Bureau decided to release many of the 2020 Census data products, including the \dhcre, using algorithms that satisfy modern privacy definitions like differential privacy \cite{dsep-dp}.

In particular, we describe a disclosure avoidance technique for the Detailed DHC-B that was designed to satisfy the following desiderata: 
\begin{itemize}
    \item \textit{Privacy:} The algorithm must satisfy a variant of differential privacy known as zero-concentrated differential privacy (zCDP) with respect to arbitrary changes of any household record's values. 
    \item \textit{Population Groups:} The algorithm must release statistics for a predefined set of race and ethnicity characteristic iterations and the following geography levels: Nation, State, County, Census Tract, Place, and AIANNH areas. The algorithm also produces statistics for equivalent geographic regions in Puerto Rico.  
    \item \textit{Adaptivity:} For each population group, the algorithm may adaptively choose which granularity of the household type and tenure tables to release. The selection of granularity depends on the T01001 total population counts of each population group. For instance, a population group with a few people may only receive a total household count (i.e., T03001 and T04001 tables), while a population group with many people may receive the full table shells (i.e., T03004 and T04002 tables). 
    \item \textit{Accuracy:} The algorithm must achieve pre-specified accuracy levels for population groups in terms of the margins of error (MOE) in output counts. Different population groups may have different MOEs specified (described later in the paper in Table~\ref{tab:moe-targets}). The MOE discussed in the paper captures error induced by disclosure avoidance alone and does not capture other sources of error such as under counts in the 2020 Census.
    \item \textit{Integrality:} The output statistics must be integers. 
    \item \textit{Minimal Consistency:} The algorithm is not required, in general, to ensure consistency. That is, different counts output by the system need not be consistent with each other (e.g., the number of households of a certain characteristic iteration in the United States need not equal the sum of the household counts for the same characteristic iteration across all states). We also note that no consistency is enforced with other data products, such as the DHC. However, some postprocessing of outputs is done to address specific demographic reasonableness concerns. These postprocessing steps are discussed in Section \ref{sec:postprocessing}.
\end{itemize}

In the rest of the paper, we describe the SafeTab-H differential privacy algorithm, discuss implementation and parameter tuning, and analyze bounds on the privacy loss achievable while satisfying the constraints mentioned above.   

\section{SafeTab-H Algorithm}
\label{sec:algorithm-description}
\safetabh is a privacy algorithm for releasing 2020 Census household type and tenure counts, iterated by detailed race and ethnicity characteristic iterations. This section covers a simplified abstraction of the SafeTab-H algorithm. Additional implementation details are discussed in Section \ref{sec:implementation}. In this section, we describe the algorithm as applied to the United States. Puerto Rico is discussed in Section \ref{sec:pr}. The algorithm acts on a private dataframe of household records derived from the 2020 Census. At a high level, the algorithm performs the following steps for each eligible population group:

\begin{enumerate}
    \item It selects a household type table variant (T03001, T03002, T03003, or T03004) and a tenure table variant (T04001 or T04002).
    \item It queries the household data to compute the basis of each selected table variant.
    \item It perturbs the computed basis of each selected table variant by adding noise drawn from a discrete Gaussian distribution.  
\end{enumerate}

The selection in the first step is fully determined by a comparison of each eligible population group's T01001 total population count to population thresholds pre-specified by the Census Bureau. The computation in the second step produces accurate counts according to the 2020 CEF. Finally, the noise infusion in the third step provides necessary privacy protection in accordance with Census Bureau policy.

For a more technical look at SafeTab-H, we begin with a description of the algorithm's input data sources and then present a pseudocode representation.

\subsection{Input Data Description}

\subsubsection{Household Data}
Household records are stored in the 2020 CEF, a relational database with multiple person and household attributes spread across several linked dataframes. Many of these attributes are irrelevant to the tabulations in the Detailed DHC-B. As such, we assume a simplified, reduced-form data representation that is sufficient for our purposes. This dataset is a private dataframe derived from the 2020 CEF that consists of a row for each household in the United States with the following attributes: BlockID, RaceEth, HouseholdType, and Tenure.

\vspace{\baselineskip}

\noindent \textbf{BlockID} is a single attribute that geolocates a household record to a unique Census block. As previously discussed, all geographic entities are aggregations of blocks. Thus, we assume that BlockID implicitly encodes each record's unique Census tract, county, and state. BlockID also encodes whether a record belongs to a place or AIANNH area and, if so, uniquely identifies its place or AIANNH area. We note that all records are vacuously included in the nation geographic entity.

\vspace{\baselineskip}

\noindent \textbf{RaceEth} is a single attribute that encodes up to eight race codes and an ethnicity code of the householder. That is, one household's RaceEth attribute may indicate the householder is Andorran and Dominican while another household's RaceEth attribute indicates a householder that is South African, Japanese, Tongan, and Not Hispanic or Latino. We assume this RaceEth conceptualization combined with Census Bureau specifications fully determines the characteristic iterations of a household.

\vspace{\baselineskip}

\noindent \textbf{HouseholdType} categorizes the relationship of the householder to other members of their household. There are five categories:
\begin{itemize}
    \item Married couple family.
    \item Other family, male householder with no spouse present.
    \item Other family, female householder with no spouse present.
    \item Nonfamily with the householder living alone.
    \item Nonfamily with the householder not living alone.
\end{itemize}

\vspace{\baselineskip}

\noindent \textbf{Tenure} categorizes the householder's owner or renter status. There are three categories:
\begin{itemize}
    \item Owned with a mortgage or a loan.
    \item Owned free and clear.
    \item Renter occupied.
\end{itemize}
\vspace{\baselineskip}

\subsubsection{Total Population Counts}
The Detailed DHC-A's T01001 table contains total population counts by detailed race and ethnicity.  These counts are used by SafeTab-H to determine the granularity of household type and tenure tabulations each population group should receive. Since the Detailed DHC-A was published before the Detailed DHC-B, for the purposes of SafeTab-H, we treat the T01001 counts as fixed exogenous inputs to the program. As with the household data, we assume a reduced-form dataframe representation. In this case, the dataframe consists of a row for each population group output by the SafeTab-P privacy algorithm. This dataset has the following attributes: PopGroup and Count.

\vspace{\baselineskip}

\noindent \textbf{PopGroup} is a single attribute that encodes the geographic entity and the characteristic iteration of a population group.

\vspace{\baselineskip}

\noindent \textbf{Count} is the population group's total population count from the Detailed DHC-A. 

\vspace{\baselineskip}

Equivalently, we can view the T01001 input data as a function that maps population groups to their corresponding total population counts. 

\vspace{\baselineskip}

\subsection{The Algorithm Description}

For reference, the notation used in this section and the algorithm pseudocode is summarized in Table~\ref{tab:algorithm-notation}. 

SafeTab-H produces tabulations for population groups. Population groups are split into sets called population group levels (specified by a geography level and an iteration level) with distinct privacy-loss budgets for each table class (household type and tenure). Records are associated with population groups via transformations that map their BlockID to geographic entities, and their RaceEth attribute to characteristic iterations. For the purposes of this section, we assume the following model for population groups:

\begin{itemize}
    \item SafeTab-H is given fixed T01001 population group counts. That is, there is a set of population groups $\mathcal{T}$ and a mapping $h: \mathcal{T} \rightarrow \mathbb{Z}$ of population groups to fixed counts of their total population. 
    
    \item SafeTab-H operates on population group levels $\mathcal{P}_1, \mathcal{P}_2, \ldots,\mathcal{P}_{\omega}$. For example, $\mathcal{P}_i$ may be the level (State, Detailed) consisting of population groups, such as (Iowa, Albanian alone) and (Kansas, German alone or in any combination).
    
    \item SafeTab-H accounts for the possibility that not all population groups in the input specification have corresponding T01001 counts. Population groups without T01001 counts do not receive household type or tenure results in the Detailed DHC-B. That is, SafeTab-H produces household type and tenure tables for each population group $P \in \mathcal{P}_i \cap \mathcal{T}$ for $1 \le i \le {\omega}$.
    
    \item In SafeTab-H, the household type (HT) and tenure (T) counts receive separate privacy-loss budgets for each population group level $\rho_1^{t}, \rho_2^{t}, \ldots, \rho_{\omega}^{t}$ with $\rho_i^{t}$ corresponding to the budget for population group level $\mathcal{P}_i$ for table class $t \in \{\text{HT, T}\}$. Privacy-loss budgets are described in greater detail in Section \ref{sec:privacy-prelim}. For now, we note that each $\rho_i^{t}$ is a positive real-valued number.
    
    \item For each $\mathcal{P}_i$, we assume we have a function $g_i: \mathcal{I} \rightarrow 2^{\mathcal{P}_i}$, where $\mathcal{I}$ is the domain of household records in the private dataframe. That is, $g_i$ maps a household record $r$ to the subset of population groups at the level $i$ to which it belongs (i.e., $g_i(r) \subset \mathcal{P}_i$). For example, suppose $i$ corresponds to the (State, Regional) level, the record $r$'s BlockID uniquely identifies its state as Idaho, and its RaceEth attribute encodes Nigerian, Beninese, Tongan, and Not Hispanic or Latino. Then $g_i(r)$ would associate the record with the following population groups: (Idaho, Sub-Saharan African alone or in any combination) and (Idaho, Polynesian alone or in any combination).
    
    \item We assume the stability of $g_i$, denoted by $\Delta(g_i)$, is known. The stability is defined as the maximum number of population groups a record could belong to in a level. Formally, $\Delta(g_i) = \max_{r \in \mathcal{I}} |g_i(r)|$ \cite{McSherry09}. Importantly, this value defines what could be the maximum based on any hypothetically possible record, rather than defining what is the maximum based on the collected 2020 Census data. In other words, the stability is a data-independent value. As described earlier, this value is $\Delta(g_i) = 9$ for all population group levels tabulated in the Detailed DHC-B.
\end{itemize}

The main algorithm is presented in Algorithm~\ref{alg:safetab-main-algorithm}.
This algorithm proceeds by looping over the population group levels.
For each population group level, we apply $g_i$ to the dataframe to map each record to the set of population groups it is associated with.
Then, for each population group in the level, we call the function \textsc{VectorizePopulationGroup}, passing in a dataframe containing just the records in that population group. However, we skip over population groups that do not have T01001 counts.

The pseudocode for the procedure \textsc{VectorizePopulationGroup} is given in Algorithm~\ref{alg:safetab-vectorize-pop-group}. This code performs the high-level steps 1 and 2 discussed earlier, returning two basis vectors (one for a household type table and another for a tenure table) for the population group. It starts by comparing the provided T01001 count against a set of given thresholds, denoted $\theta_1, \theta_2,$ and $\theta_3$ for the HT table class and $\psi_1$ for the T table class. Depending on which thresholds the published count exceeds, it selects an HT table variant (T03001, T03002, T03003, or T03004) and a T table variant (T04001 or T04002). It then computes the basis vectors for the selected household type and tenure tables by counting how many households in the population group belong to each component of the respective basis vectors.

The computed HT vectors of each population group in the population group level are stacked to create one large HT vector for the population group level. Similarly, one large T vector is constructed.  

The HT and T vectors for the population group level are then passed through the \textsc{VectorDiscreteGaussian} procedure using the level's privacy-loss budget for the respective table classes.

The pseudocode for the procedure \textsc{VectorDiscreteGaussian} is given in Algorithm~\ref{alg:base-discrete-gaussian}. This procedure carries out the high-level step 3 by adding independently drawn noise from the discrete Gaussian distribution to each of the input vector's components. The distribution is scaled according to the privacy-loss budget input to the procedure.

Next, we introduce background material on zCDP so that we can analyze the privacy guarantees of the SafeTab-H algorithm.
\begin{table}[t]
    \centering
    \begin{tabular}{c p{.8\linewidth}}
        \toprule
        \textbf{Notation} & \textbf{Description} \\
        \midrule
        $\omega$ & the number of population group levels \\
        $\mathcal{P}_i$ & population group level $i$  \\
        $\mathcal{T}$ & population groups with T1 counts \\
        $h$ & a mapping of population groups to their T1 counts \\
        $\rho_i^t$ & the privacy-loss budget allocated to population group level $i$ for table $t \in \{\text{HT, T}\}$ \\
        $g_i$ & a function mapping records to the set of population groups in $\mathcal{P}_i$ to which the record belongs \\
        $\Delta(g_i)$ & $\max_{r \in \mathcal{I}} |g_i(r)|$ \\
        $[a, b]$ & The range of integers starting with $a$ and ending with $b$, inclusive.
    \end{tabular}
    \caption{A summary of the notation used in Section~\ref{sec:algorithm-description}}
    \label{tab:algorithm-notation}
\end{table}

\begin{algorithm}[t]
\caption{\label{alg:safetab-main-algorithm} The main \safetabh algorithm.}
\begin{algorithmic}[1]
\Require
$df$: A private dataframe with attributes [BlockID, RaceEth, HouseholdType, Tenure] and one row for each household in the United States.
\Require $\{\rho_i^t\}_{i \in [1,\omega]}$: Privacy-loss parameters for each population-group level $i \in [1, \omega]$ and table class $t \in \{\text{HT, T}\}$.
\Require $(\mathcal{T}, h)$: A mapping $h$ of population groups to their T01001 total population counts, along with the mapping's domain $\mathcal{T}$. 

\Procedure{SafeTab-H}{$df$, $\{\rho_i^t\}$, $\mathcal{T}$, $h$}

\For{$i \in [1, \omega]$} \label{line:pop-group-level-loop}
\State $df_i \leftarrow df$.flatmap($g_i$);
\Comment{$df_i$ has schema [PopGroup, HouseholdType, Tenure]}
\State $s \leftarrow \Delta(g_i)$ 
\Comment{1 row in $df$ may result in $\leq s$ rows in $df_i$}
\State $V_{HT},V_{T} \gets [\ ]$ \Comment{Initialize empty vectors of counts}
\For{$P \in \mathcal{P}_i$ \label{line:iteration-loop}}
\If{$P$ $\notin$ $\mathcal{T}$}
\State \textbf{continue}
\Comment{Skip population groups that do not have T01001 counts}
\EndIf

\State $c \leftarrow h(P)$
\Comment{Get the T01001 count of population group $P$}

\State $df_P$ $\leftarrow$ $df_i$.filter(PopGroup $== P$)
\State $v_{HT},v_{T} \leftarrow$\Call{VectorizePopulationGroup}{$df_p$, $P$, $c$}
\State $V_{HT}$.append($v_{HT}$) 
\State $V_{T}$.append($v_{T}$)
\EndFor
\State \textbf{Output} \Call{VectorDiscreteGaussian}{$V_{HT}, \rho_i^{HT}/s$} \label{line:gaussian-t3}
\State \textbf{Output} \Call{VectorDiscreteGaussian}{$V_{T}, \rho_i^{T}/s$} \label{line:gaussian-t4}
\EndFor
\EndProcedure
\end{algorithmic}
\end{algorithm}

\begin{algorithm}[t]
\caption{\label{alg:safetab-vectorize-pop-group} Subroutine of \safetabh that returns a vector of HT counts and a vector of T counts for a single population group.}
\begin{algorithmic}[1]
\Require $df$: A private dataframe with attributes [PopGroup, HouseholdType, Tenure]. This dataframe should only contain household records in population group $P$.
\Require $P$: The population group.
\Require $c$: The T01001 count of population group $P$.

\Procedure{VectorizePopulationGroup}{$df, P, c$}
\State \textit{// Adaptive process for Household Type tabulations}\label{line:t3-start}
\State $v_{HT}, v_{T} \leftarrow [\ ]$
\If{$c > \theta_3$}
\State \textit{// Vectorize basis of T03004}
\State $v\leftarrow$ df.map(HouseholdType $\rightarrow$ T03004 basis).groupby(T03004 basis).count()
\State $v_{HT}$.append($v$) \label{line:t03004}

\ElsIf{$c > \theta_2$}
\State \textit{// Vectorize basis of T03003}
\State $v\leftarrow$ df.map(HouseholdType $\rightarrow$ T03003 basis).groupby(T03003 basis).count()
\State $v_{HT}$.append($v$) \label{line:t03003}

\ElsIf{$c > \theta_1$}

\State \textit{// Vectorize basis of T03002}
\State $v\leftarrow$ df.map(HouseholdType $\rightarrow$ T03002 basis).groupby(T03002 basis).count()
\State $v_{HT}$.append($v$)\label{line:t03002}

\Else
\State \textit{// Vectorize basis of T03001}
\State $v_{HT}$.append(df.count())\label{line:t03001} 
\EndIf\label{line:t3-end}

\State \textit{// Adaptive process for Tenure tabulations}\label{line:t4-start}

\If{$c > \psi_1$}

\State \textit{// Vectorize basis of T04002}
\State $v\leftarrow$ df.map(Tenure $\rightarrow$ T04002 basis).groupby(T04002 basis).count()
\State $v_{T}$.append($v$)\label{line:t04002}

\Else
\State \textit{// Vectorize basis of T04001}
\State $v_{T}$.append(df.count()) \label{line:t04001} 
\EndIf\label{line:t4-end}
\State \textbf{Return} $v_{HT}, v_{T}$
\EndProcedure
\end{algorithmic}
\end{algorithm}

\begin{algorithm}[t]
\caption{\label{alg:base-discrete-gaussian} The discrete Gaussian mechanism for vectors.}
\begin{algorithmic}[1]
\Require $a$: An $n$ dimensional vector of integers.
\Require $\rho$: A privacy-loss parameter.
\Procedure{VectorDiscreteGaussian}{$a, \rho$}
\State $y \gets \mathcal{N}^{n}_{\mathbb{Z}}\left(\frac{1}{2\rho}\right)$
\State \textbf{return} $a + y$
\EndProcedure
\end{algorithmic}
\end{algorithm}

\clearpage


\section{Privacy Preliminaries}
\label{sec:privacy-prelim}

In this section, we give necessary background on zCDP and the  privacy properties that it guarantees.

\subsection{Privacy Definitions}
\label{sec:privacy-defintions}

\begin{definition}[Neighboring Databases]
\label{def:neighboring-databases}
Let $x,x'$ be databases represented as multisets of tuples. We say that $x$ and $x'$ are \emph{neighbors} if their symmetric difference is 1.
\end{definition}

\begin{definition}[Bounded-Neighboring Databases]
\label{def:bounded-neighboring-databases}
Let $x,x'$ be databases represented as multisets of tuples. We say that $x$ and $x'$ are \emph{bounded neighbors} if they differ by \emph{arbitrarily changing} at most one tuple.
\end{definition}
We sometimes refer to neighboring databases as ``unbounded neighbors'' to differentiate between Definitions~\ref{def:neighboring-databases} and \ref{def:bounded-neighboring-databases}. However, it should be assumed we are referring to Definition~\ref{def:neighboring-databases} unless a distinction is explicitly made. \par

The mathematical definition of zCDP expresses a bound on the \emph{R\'enyi divergence} between the distributions of a mechanism run on neighboring databases.

\begin{definition}
The \emph{R\'enyi divergence of order $\alpha$} between distribution $P$ and distribution $Q$, denoted $D_{\alpha}(P \| Q)$, is defined as
\begin{equation}
    D_{\alpha}(P \| Q) = \frac{1}{\alpha-1}\log\left(\underset{x \sim P}{\mathbb{E}} \left[ \left( \frac{P(x)}{Q(x)} \right)^{\alpha-1}\right]\right)
\end{equation}
\end{definition}

\begin{definition}(zCDP \cite{BunS16})\label{def:zcdp}
An algorithm $M: \mathcal{X} \rightarrow \mathcal{Y}$ satisfies $\rho$-zCDP if for all neighboring $x, x' \in \mathcal{X}$ and for all $\alpha \in (1, \infty)$,
\begin{equation}
    D_{\alpha}(M(x) \| M(x')) \le \rho \alpha.
\end{equation}
\end{definition}
We also define bounded zCDP, which considers bounded-neighboring databases instead of unbounded-neighboring databases.
\begin{definition}(Bounded zCDP \cite{BunS16})\label{def:bounded-zcdp}
An algorithm $M: \mathcal{X} \rightarrow \mathcal{Y}$ satisfies $\rho$-zCDP if for all bounded neighbors $x, x' \in \mathcal{X}$ and for all $\alpha \in (1, \infty)$,
\begin{equation}
    D_{\alpha}(M(x) \| M(x')) \le \rho \alpha.
\end{equation}
\end{definition}

\subsection{Privacy Properties}

\subsubsection{Composition}
\label{sec:composition}

One of the most useful and important properties of privacy definitions is their behavior under composition. For the SafeTab-H privacy analysis, sequential composition for zCDP is sufficient.

\begin{lemma}(Adaptive sequential composition of zCDP \cite{BunS16})
\label{lem:sequential-composition-zcdp}
Let $M_1: \mathcal{X} \rightarrow \mathcal{Y}$ and $M_2: \mathcal{X} \times \mathcal{Y} \rightarrow \mathcal{Z}$ be mechanisms satisfying $\rho_1$-zCDP and $\rho_2$-zCDP respectively. Let $M_3(x) = M_2(x, M_1(x))$. Then $M_3$ satisfies $(\rho_1 + \rho_2)$-zCDP.
\end{lemma}

\subsubsection{Postprocessing}
\label{sec:post-processing-background}

Zero-concentrated differential privacy is closed under postprocessing, meaning that the privacy guarantee cannot be weakened by manipulating the outputs of a zCDP mechanism without reference to the protected inputs.

\begin{lemma}(Postprocessing for zCDP \cite{BunS16})
\label{lem:post-processing}
Let $M: \mathcal{X} \rightarrow \mathcal{Y}$ and $f:\mathcal{Y} \rightarrow \mathcal{Z}$ be randomized algorithms. Suppose $M$ satisfies $\rho$-zCDP. Then $f \circ M: \mathcal{X} \rightarrow \mathcal{Z}$ satisfies $\rho$-zCDP.
\end{lemma}

\subsection{Base Mechanism}
\label{sec:base-mechansisms}

\begin{definition}[L2 Sensitivity]
\label{def:sensitivity}
Given a vector function $q: \mathcal{X} \rightarrow \mathbb{Z}^n$, the sensitivity of $q$ is \linebreak $\sup_{x, x'} \|q(x)-q(x')\|_2$ where $x$ and $x'$ are neighboring databases and $\| \cdot \|_2$ is the Euclidean norm.
\end{definition}

There is an equivalent notion of bounded sensitivity.

\begin{definition}[Bounded L2 Sensitivity]
\label{def:bounded_sensitivity}
Given a vector function $q: \mathcal{X} \rightarrow \mathbb{Z}^n$, the sensitivity of $q$ is $\sup_{x ,x'} \|q(x)-q(x')\|_2$ where $x$ and $x'$ are bounded-neighboring databases and $\| \cdot \|_2$ is the Euclidean norm.
\end{definition}

\begin{definition}
\label{def:discrete-gaussian-distribution}
The discrete Gaussian distribution $\mathcal{N}_{\mathbb{Z}}(\sigma^2)$ centered at 0 is
\begin{equation}
\forall x \in \mathbb{Z}, \quad \Pr[X=x] = \frac{e^{-x^2/2\sigma^2}}{\sum_{y \in \mathbb{Z}}e^{-y^2/2 \sigma^2}}.
\end{equation}
\end{definition}

\begin{lemma}{\cite{CanonneK2020}}
\label{lem:discrete-gaussian-satisfies-zcdp}
Let $q: \mathcal{X} \rightarrow \mathbb{Z}^n$ with sensitivity $\Delta$. 
Then outputting \textsc{VectorDiscreteGaussian}$(q(x), \rho)$ from Algorithm~\ref{alg:base-discrete-gaussian} satisfies $\Delta^2\rho$-zCDP. 
\end{lemma}


\section{SafeTab-H Privacy Analysis}
\label{sec:safetab-discrete-gaussian-privacy}

In this section, we show that the SafeTab-H algorithm presented in Section~\ref{sec:algorithm-description} satisfies zero-concentrated differential privacy (zCDP). We then demonstrate that \safetabh also satisfies bounded zCDP with a privacy loss that increases by a factor of 2.

\begin{theorem}
\label{thm:safetab-satisfies-zcdp}
Let $\rho_{total} = \sum_{i=1}^{\omega} \rho_i^{HT} + \rho_i^{T}$. Algorithm~\ref{alg:safetab-main-algorithm} satisfies $\rho_{total}$-zCDP with respect to its inputs.
\end{theorem}

\begin{proof}
The proof of Theorem~\ref{thm:safetab-satisfies-zcdp} follows via the combination of sensitivity analysis along with the fact that the base mechanism, \textsc{VectorDiscreteGaussian}, satisfies zCDP.

We first argue the transformation \textproc{VectorizePopulationGroup} in Algorithm~\ref{alg:safetab-vectorize-pop-group} has sensitivity $1$ with respect to each of the outputs $v_{HT},v_{T}$.

We decompose our argument into two  subclaims.

\textbf{Subclaim 1:} The vector, $v_{HT}$, produced in Lines \ref{line:t3-start}-\ref{line:t3-end} of \textproc{VectorizePopulationGroup} can have at most one of its elements increased by $1$ due to the addition of a record or decreased by $1$ due to the removal of a record.

\textbf{Proof of subclaim 1:} Assume the population group $P$ and its corresponding T01001 count $c$ are fixed.

If $c > \theta_1$, there are 3 similar cases. We show the case where $c > \theta_3$ and omit the others. 

In \textproc{VectorizePopulationGroup}, each record in the input dataframe for population group $P$ maps to exactly one group of the T03004 basis groupby. Since $v_{HT}$ is a vector of the counts of each group, it follows that the addition of a single record can increase one element of $v_{HT}$ by at most $1$ or, alternatively, the removal of a single record can decrease one element of $v_{HT}$ by at most $1$.

If $c \leq \theta_1$, the vector $v_{HT}$ is a size $1$ vector containing only the total count. It follows that the addition of a single record can increase the one element of $v_{HT}$ by at most $1$. Alternatively, the removal of a single record can decrease the one element of $v_{HT}$ by at most $1$.

This proves subclaim 1.

\textbf{Subclaim 2:} The vector, $v_{T}$, produced in Lines \ref{line:t4-start}-\ref{line:t4-end} of \textproc{VectorizePopulationGroup} can have at most one of its elements increased by $1$ due to the addition of a record or decreased by $1$ due to the removal of a record. We omit the proof, as it is nearly identical to that of subclaim 1.

Next, we claim that the $i$th loop of the \textbf{for} loop on line~\ref{line:pop-group-level-loop} of Algorithm~\ref{alg:safetab-main-algorithm} satisfies ($\rho_i^{HT}+\rho_i^{T}$)-zCDP. By the definition of $s$, any particular record can appear in the input ($df_P$) in at most $s$ calls to \textsc{VectorizePopulationGroup}. The vectors $V_{HT},V_{T}$ are the combination of each of the vectors produced by \textsc{VectorizePopulationGroup}. Since a single record can appear in $s$ calls to \textsc{VectorizePopulationGroup} and each call can change a single element by at most $1$, the addition or removal of a single record can change at most $s$ elements in $V_{HT}, V_{T}$ by at most $1$. Therefore $\|V_{HT}-V'_{HT}\|_2 \leq \sqrt{s}$, where $V_{HT}$ and $V_{HT}'$ are data vectors resulting from neighboring databases. The same holds for the vector $V_{T}$. Thus, by Lemma~\ref{lem:discrete-gaussian-satisfies-zcdp}, the call to \textsc{VectorDiscreteGaussian} with privacy-loss parameter $\rho_i^{HT}/s$ on Line~\ref{line:gaussian-t3} satisfies $\rho_i^{HT}$-zCDP. Similarly, the call on Line~\ref{line:gaussian-t4} with privacy-loss parameter $\rho_i^{T}/s$ satisfies $\rho_i^{T}$-zCDP. By sequential composition (Lemma~\ref{lem:sequential-composition-zcdp}), the combination satisfies ($\rho_i^{HT}+\rho_i^{T}$)-zCDP.

Finally, the overall algorithm satisfies $(\sum_{i=1}^{\omega} \rho_i^{HT} + \rho_i^{T})$-zCDP by Lemma~\ref{lem:sequential-composition-zcdp}.
\end{proof}

\subsection{Converting to Bounded zCDP}
\label{sec:bounded-dp}

\begin{theorem}
\label{thm:safetab-satisfies-bounded-zcdp}
Let $\rho_{total} = \sum_{i=1}^{\omega} \rho_i^{HT} + \rho_i^{T}$. Algorithm~\ref{alg:safetab-main-algorithm} satisfies bounded $2\rho_{total}$-zCDP with respect to its inputs.
\end{theorem}

\begin{proof}
The proof of Theorem~\ref{thm:safetab-satisfies-bounded-zcdp} follows via the combination of sensitivity analysis along with the fact that the base mechanism, \textsc{VectorDiscreteGaussian}, satisfies zCDP.  Under bounded zCDP, it is sufficient to bound the changes due to the removal of one record \textit{and} addition of another. As a result, we reason specifically about the effects of adding and removing a record separately and the effects when both happen.

We first argue the transformation \textproc{VectorizePopulationGroup} in Algorithm~\ref{alg:safetab-vectorize-pop-group} has  sensitivity $1$ with respect to each of the outputs $v_{HT},v_{T}$. 

We decompose our argument into two  subclaims. Namely, 

\textbf{Subclaim 1:} The vector, $v_{HT}$, produced in Lines \ref{line:t3-start}-\ref{line:t3-end} of \textproc{VectorizePopulationGroup} can have at most one of its elements increased by $1$ due to the addition of a record or decreased by $1$ due to the removal of a record.

\textbf{Proof of subclaim 1:} Assume the population group $P$ and its corresponding T01001 count $c$ are fixed.
If $c > \theta_1$, there are 3 similar cases. We show the case where $c > \theta_3$ and omit the others. 

In \textproc{VectorizePopulationGroup}, each record in the input dataframe for population group $P$ maps to exactly one group of the T03004 basis groupby. Since $v_{HT}$ is a vector of the counts of each group, it follows that the addition of a single record can increase one element of $v_{HT}$ by at most $1$ or, alternatively, the removal of a single record can decrease one element of $v_{HT}$ by at most $1$.

If $c \leq \theta_1$, the vector $v_{HT}$ is a size $1$ vector containing only the total count. It follows that the addition of a single record can increase the one element of $v_{HT}$ by at most $1$. Alternatively, the removal of a single record can decrease the one element of $v_{HT}$ by at most $1$.

This proves subclaim 1.

\textbf{Subclaim 2:} The vector, $v_{T}$, produced in Lines \ref{line:t4-start}-\ref{line:t4-end} of \textproc{VectorizePopulationGroup} can have at most one of its elements increased by $1$ due to the addition of a record or decreased by $1$ due to the removal of a record. We omit the proof, as it is nearly identical to that of Subclaim 1.

Next, we claim that the $i$th loop of the \textbf{for} loop on line~\ref{line:pop-group-level-loop} of Algorithm~\ref{alg:safetab-main-algorithm} satisfies bounded ($2\rho_i^{HT}+2\rho_i^{T}$)-zCDP. By the definition of $s$, any particular record can appear in the input ($df_P$) in at most $s$ calls to \textsc{VectorizePopulationGroup}. The vectors $V_{HT},V_{T}$ are the combination of each of the vectors produced by \textsc{VectorizePopulationGroup}. Since a single record can appear in $s$ calls to \textsc{VectorizePopulationGroup} and each call can change a single element by at most $1$, the addition or removal of a single record can change at most $s$ elements in $V_{HT}, V_{T}$ by at most $1$. Since the addition of one record can increase at most $s$ elements by $1$ and the removal of one record can decrease at most $s$ elements by $1$, we have $\|V_{HT}-V'_{HT}\|_2 \leq \sqrt{2s}$, where $V_{HT}$ and $V_{HT}'$ are data vectors resulting from bounded-neighboring databases. The same holds for the vector $V_{T}$. Thus, by Lemma~\ref{lem:discrete-gaussian-satisfies-zcdp}, the call to \textsc{VectorDiscreteGaussian} with privacy-loss parameter $\rho_i^{HT}/s$, on Line~\ref{line:gaussian-t3} satisfies bounded $2\rho_i^{HT}$-zCDP. Similarly, the call on Line~\ref{line:gaussian-t4} with privacy-loss parameter $\rho_i^{T}/s$ satisfies bounded $2\rho_i^{T}$-zCDP. By sequential composition (Lemma~\ref{lem:sequential-composition-zcdp}), the combination satisfies bounded ($2\rho_i^{HT}+2\rho_i^{T}$)-zCDP.

Finally, the overall algorithm satisfies bounded $(\sum_{i=1}^{\omega} 2\rho_i^{HT} + 2\rho_i^{T})$-zCDP by Lemma~\ref{lem:sequential-composition-zcdp}.
\end{proof}

\section{Implementation of SafeTab-H}
\label{sec:implementation}

The algorithm presented in Section \ref{sec:algorithm-description} is a simplified version of the implemented SafeTab-H program. In this section, we describe some of the differences between the implementation and the simplified algorithm. We focus on differences that could affect the privacy calculus and describe why the implementation is equivalent to the simplified algorithm.

\subsection{Input Validation}\label{sec:validation}
Input validation is an important step before deploying a differentially private algorithm. SafeTab-H performs extensive validation of its input data to ensure that provided data are in the expected formats and are internally consistent. There is not an impact to any privacy guarantees when input validation is done for the public datasets, like the list of all geographic entities for which data is to be tabulated. The input validation on private datasets that contain the data of individual census respondents has privacy implications because validation failures can reveal information about the dataset. However, validation failures are only visible to the trusted curator running the program, and on failure, no part of the differentially private program is run. Validation failures are made available to the trusted curator, so they can correct any errors in the provided input files before executing the differentially private program. Failures in validation are not released publicly and therefore, do not contribute to any privacy loss. 

\subsection{Tumult Analytics}\label{sec:analytics}
Rather than directly calculating stability and sampling from noise distributions, SafeTab-H is implemented using Tumult Analytics\cite{berghel2022tumult}, a framework for implementing differentially private queries. 

A key benefit of using Tumult Analytics is that all access to the sensitive data is mediated through a Tumult Analytics \texttt{Session}. The \texttt{Session} tracks all the transformations and measurements performed on the sensitive data and is able to correctly compute the total privacy loss of the computation on the sensitive data. In SafeTab-H, we construct an Analytics \texttt{Session} with:
\begin{itemize}
    \item The total privacy-loss budget for the pipeline (calculated as the sum of all budgets $\rho_{i}$).
    \item The private dataset of household records.
    \item The public datasets (information on all detailed race and ethnicity groups, characteristic iterations, geographic entities, and total population counts for each population group).
    \item The neighboring definition (privacy is with respect to the addition/removal of one record from the private dataset). Note that Tumult Analytics does not support bounded neighbors.
    \item The privacy definition to be satisfied (e.g., zCDP).
\end{itemize}
We then implement all data transformations (like mapping a record to its characteristic iterations) and queries within the framework. Tumult Analytics tracks the stability throughout the transformations and applies an appropriate amount of noise to the final queries, guaranteeing that the outputs are differentially private and no more than the total budget is expended.

The use of Tumult Analytics allows (and necessitates) some other deviations from the simplified algorithm. Rather than vectorizing each population group sequentially in a for-loop, we use Analytics' \texttt{groupby} feature to tabulate many population groups at once. Analytics requires users of the \texttt{groupby} feature to specify all groups to tabulate in advance in the form of a \texttt{KeySet} object. In the case of SafeTab-H, we construct \texttt{KeySets} containing the combinations of possible geographic areas, characteristic iterations, and either household types or tenure categories. We build these \texttt{KeySets} using the T01001 total population counts output by SafeTab-P (rather than relying on observed groups present in the 2020 CEF). \texttt{KeySets} do not use the confidential 2020 CEF data to ensure that we do not reveal whether population groups are empty via their presence or absence in the output data. We note the simplified algorithm also allowed for the possibility that $df_P$ (the filtered dataframe containing records associated with population group $P$) is empty.

\subsection{Postprocessing for Addressing Demographic Reasonableness Concerns}\label{sec:postprocessing}
After the differentially private algorithm has completed, we perform several additional postprocessing steps. Because these steps are purely postprocessing, they cannot affect the differential privacy guarantees per Lemma \ref{lem:post-processing}. These postprocessing steps are designed to address specific data quality concerns that arise when adding noise to tabular statistics. All postprocessing was implemented under the direction of subject-matter experts at the Census Bureau. These steps do not exhaustively address all possible demographic reasonableness concerns. 

\subsubsection{Marginals}

As mentioned in Section~\ref{sec:algorithm-description}, SafeTab-H selects different table variants for each population group based on their T01001 counts and then produces noisy estimates of the selected table's basis. For example, if T03003 were selected for a population group, the pseudocode would produce estimates for the categories "Married Couple Family", "Other Family", "Nonfamily: Householder Living Alone", and "Nonfamily: Householder Not Living Alone" but would not produce values for "Family Household", "Nonfamily Household", and "Total". To ensure that the selected table variant is released, SafeTab-H has a postprocessing step that aggregates each table basis to construct the complete table shell. In our example, "Married Couple Family" and "Other Family" are summed to produce a noisy estimate for the "Family Household" table cell, "Nonfamily: Householder Living Alone" and "Nonfamily: Householder Not Living Alone" are summed to produce a noisy estimate for the "Nonfamily Household" table cell, and all four basis values are summed to produce a noisy estimate for the "Total" table cell. This ensures the complete T03003 table shell is output by SafeTab-H and that the table marginals for a population group are consistent with the table's basis.

\subsubsection{Suppression}
SafeTab-H does not implement its own suppression. However, SafeTab-H does not produce statistics for any population group suppressed by SafeTab-P (because SafeTab-H relies on the population group totals that have undergone suppression). Thus, some population groups that might otherwise be expected to appear in the output are missing. These suppressed counts are not published in any fashion in the S-DHC.

Suppressing outputs based on the noisy counts produced by SafeTab-P (as opposed to the noisy counts calculated by SafeTab-H) does not introduce bias into the counts that are directly published within the Detailed DHC-B. Additionally, the "suppressed" population groups are completely deterministic with respect to a fixed SafeTab-P output, so there is no randomness in the process.

\subsubsection{Coterminous Geographies}
Sometimes, two or more geographic entities in different geographic summary levels share the same geographic boundaries (i.e., they are aggregated from identical collections of Census blocks). For example, Washington, D.C. is tabulated as a state, county, and place. These geographic entities are called \emph{coterminous}. Another coterminous example is a county containing a single Census tract. A characteristic iteration receives different independent noisy measurements for each geographic summary level of a given coterminous area. However, its counts should be identical at each summary level. Some geographic entities at different summary levels that do not share the same geographic boundaries should still be statistically equivalent. For example, if a county contains one water-only tract and one nonwater tract, characteristic iterations should have the same counts in the nonwater tract as in the county. We also consider statistically equivalent geographic areas to be coterminous. Subject-matter experts at the Census Bureau created a postprocessing step (implemented as a standalone program) that corrects inconsistencies in coterminous geographic entities. Since this correction is performed outside the DAS, it is out of scope for this paper, but we note that as pure postprocessing, it does not impact the privacy analysis.

\subsubsection{Tabulation System Suppression}
Additional demographic reasonableness corrections are addressed outside the DAS. In particular, the Decennial Tabulation System also performs suppression postprocessing. Again, per Lemma \ref{lem:post-processing}, this does not impact the privacy analysis. The suppression conducted by non-DAS systems is out of scope for this paper but includes further enforcement of nonnegativity as well as suppression in cases where noisy counts for alone characteristic iterations appear to be greater than the corresponding noisy counts for alone or in any combination characteristic iterations. These suppressed counts are published as "X" in the S-DHC. 

\subsection{Input Sourcing}
Section \ref{sec:algorithm-description} describes the total population counts as being sourced by the Detailed DHC-A. In reality, this data is sourced by the outputs of the SafeTab-P privacy program that includes some suppression of small counts and postprocessing for coterminous geography consistency but does not include Decennial Tabulation System suppressions. As a result, the SafeTab-H algorithm receives total population counts for more population groups than those released publicly in the Detailed DHC-A. Discrepancies between the counts input to SafeTab-H and the counts available in the Detailed DHC-A are strictly due to postprocessing of SafeTab-P's outputs. Importantly, these differences do not degrade the stated privacy guarantees of either SafeTab-P or SafeTab-H. Also, it should be noted that population groups that do appear in SafeTab-H's total population inputs and in the Detailed DHC-A have identical counts in both sources. 

Because the DHC releases total housing unit counts without noise infusion, SafeTab-H does not produce statistics for geographic areas where no occupied or vacant housing units exist. This is ensured by removing such geographic areas from the input data before the SafeTab-H algorithm begins its processing. 

\subsection{Other Implementation Details}\label{sec:complication}
We note a few other implementation differences that are primarily driven by the specification requirements of the system and data wrangling aspects of the code. These include the function mapping detailed race and ethnicity codes to characteristic iterations, rules for what statistics are tabulated for different population groups, and the handling of Puerto Rico. 

\subsubsection{Mapping Detailed Race and Ethnicity Codes to Characteristic Iterations}
The pseudocode in Section \ref{sec:algorithm-description} abstracts the process of mapping a household's input record into their corresponding population group as a function $g_{i}$. In practice, this process requires joining against several specification input files and some subtle logic (to determine whether a household qualifies for an alone race characteristic iteration in addition to an alone or in any combination race characteristic iteration). However, the result is functionally equivalent to the $g_{i}$ abstraction - each record is mapped to a number of geographic entities and characteristic iterations. The stability factor of the implemented transformations, the equivalent of $\Delta(g_i)$, is automatically tracked by Tumult Analytics rather than being computed by hand. 

The pseudocode also ignores the logic associated with pre-processing a specified universe of geographic entities and iteration codes into population group levels $\mathcal{P}_i$. That is, the master list of all population groups divided into population group levels is constructed through a combination of specification files rather than being handed directly to the system.

\subsubsection{Puerto Rico}\label{sec:pr}
The SafeTab-H algorithm presented in Algorithm \ref{alg:safetab-main-algorithm} describes an input dataframe consisting of records of every household in the United States. However, the same algorithm is applied to a dataframe consisting of records from Puerto Rico. In implementation, SafeTab-H tabulates data for the United States and Puerto Rico in two separate passes.

\section{Parameters and Tuning}\label{sec:params}

Between the pseudocode representation of SafeTab-H and the selected implementation details presented earlier, we have alluded to a number of parameters that must be set before executing a run of the SafeTab-H program. Parameters are adjustable factors that must be fixed to fully define the nature of the program (e.g., the noise scale employed in \textsc{VectorDiscreteGaussian}, the privacy-loss budgets for population group levels, and total population thresholds to qualify for different household type and tenure table variants). With regard to SafeTab-H specifically (but any differentially private algorithm generally), parameter selection is a matter of policy. The Census Bureau's Data Stewardship Executive Policy (DSEP) committee, in consultation with subject-matter experts as well as internal and external privacy experts, approved all available parameters for the Detailed DHC-B. Parameter selection necessitates trade-offs, as many of these parameters are dependent on each other. To illustrate, we consider a fundamental relationship between the privacy-loss parameters and their corresponding margins of error with discrete Gaussian noise distributions.

\subsection{Error Bounds}\label{sec:discrete-gauss-error-bound}
SafeTab-H was designed to have predictable, tunable error so that accuracy targets of Census Bureau can be achieved with a known probability. Here, we examine the utility of Algorithm \ref{alg:safetab-main-algorithm} with discrete Gaussian noise. 

\begin{definition}
The 95\% MOE is half the width of the 95\% confidence interval.
\end{definition}

Since 95\% is the only confidence interval considered in this paper, we often write MOE without the 95\% qualifier.

We begin by stating a portion of Proposition 25 from \cite{CanonneK2020}.

\begin{proposition}[Proposition 25 in \cite{CanonneK2020}]\label{lem:discrete-gaussian-bound}

For all $m \in \mathbb{Z}$ with $m \geq 1$, and for all $\sigma \in \mathbb{R}$ with $\sigma > 0$, $\Pr[X \geq m]_{X \leftarrow \mathcal{N}_{\mathbb{Z}}(\sigma^2)} \leq \Pr[X \geq m-1]_{X \leftarrow \mathcal{N}(\sigma^2)}$.

This says that discrete Gaussian tails are tighter than the corresponding continuous Gaussian tails, which follows because the discrete Gaussian is sub-Gaussian and has tighter variance. 

\end{proposition}

The following corollary extends the tail bounds to the real numbers. 
\begin{corollary}
For all $x, \sigma \in \mathbb{R}$ with $x \geq 1$ and $\sigma > 0$,  $\Pr[X > x]_{X \leftarrow \mathcal{N}_{\mathbb{Z}}(\sigma^2)} \leq \Pr[X > \lfloor x \rfloor]_{X \leftarrow \mathcal{N}(\sigma^2)}$.
\end{corollary}

Figure 2 of \cite{CanonneK2020} provides an intuitive visualization of these tail bounds. Using the continuous Gaussian to upper bound MOE in the discrete Gaussian, it follows that the discrete Gaussian has $X \in [-\lfloor 1.96\sigma \rfloor, \lfloor 1.96\sigma \rfloor]$ with probability of at least 95\%.
That is, $MOE \leq \lfloor 1.96\sigma \rfloor$.  

Recall that $\sigma^2 = \frac{1}{2\rho}$ in Algorithm \ref{alg:base-discrete-gaussian}. Combining these two equations and solving for $\rho$, yields the following result.

\begin{corollary}\label{cor:dgauss-rho-from-moe}
The base discrete Gaussian mechanism run with $\rho^t = \frac{1.92}{\lfloor MOE \rfloor^2}$ has a 95\% margin of error of at most $MOE$.
\end{corollary}

For a population group in level $i$, the MOE in any single directly computed estimate in Algorithm \ref{alg:safetab-vectorize-pop-group}  is $\left \lfloor1.96\sqrt{\frac{s}{2\rho_i^t}} \right \rfloor$ where $t \in \{\text{HT, T}\}$ and $s$ is the stability of the function that maps records to their corresponding population groups in population group level $i$. 

\subsection{Parameter Identification, Trade-offs, and Outcomes}

Parameters tend to impact some combination of these three aspects: data confidentiality, data accuracy, and data availability. Data availability refers to the volume of tabular statistics released. Data accuracy is primarily measured by the MOEs of the tabular statistics. Data confidentiality refers to privacy loss, and it is measured by the $\rho$-zCDP parameters of the algorithm. For example, excluding population group level $i$ would reduce data availability but improve privacy since the privacy-loss budgets $\rho_i^{\text{HT}}$ and $\rho_i^{\text{T}}$ are no longer necessary.  

To aid in the understanding of these trade-offs, we created the SafeTab-H Analysis Tool to provide hands-on experience exploring the parameter space. First, we provide a brief summary of this tool. Then, we highlight specific parameters and the critical decisions made by the DSEP committee based on recommendations from subject-matter expert or DAS scientists. Several of the subject-matter expert recommendations were influenced by interaction with the SafeTab-H Analysis tool.

\subsection{Parameter Tuning Using the SafeTab-H Analysis Tool}

The SafeTab-H Analysis Tool is an easy-to-use interactive decision support tool. The tool, implemented in the Microsoft Excel program, allowed users to interactively specify: 
\begin{itemize}
\item The set of geography levels and characteristic iteration levels that constitute the universe of population groups for which statistics are tabulated. 
\item The maximum number of races a person is associated with to an integer in the range 1-8.
\item Expected MOE targets with adjustable confidence levels.
\end{itemize}
Based on these parameters, the tool computed required privacy-loss parameters to achieve the desired target error.  The computations were performed using analytical formulae for expected error of noise mechanisms employed in the SafeTab-H algorithm. 

In the next section, we describe in the next section some of the key parameters considered and the decision process used to set these parameters.

\subsubsection{Parameter Selection}

\noindent \emph{\textbf{Population group levels and population thresholds}}: As a reminder, a population group level is defined by a geographic summary level, such as Nation, State, and County, and an iteration level (i.e., Detailed or Regional). The SafeTab-H Analysis Tool helped subject-matter experts understand the privacy-loss budget required to produce acceptably accurate statistics at each geography level. The Census Bureau also gathered feedback from data users on these topics and ultimately settled on the levels referenced in Table \ref{tab:moe-targets}. The population thresholds selected for Detailed DHC-B are similar to the population thresholds in Detailed DHC-A.

\vspace{\baselineskip}

\begin{table}[t]
    \centering
    \begin{tabular}{c c c c c c}
    \toprule
    Population Group Level & MOE Target &  
    \multicolumn{2}{c}{Unbounded Privacy Loss} & \multicolumn{2}{c}{Bounded Privacy Loss}
    \\
    \cmidrule(lr{.75em}){3-4}
    \cmidrule(lr{.75em}){5-6}
    & & Household Type & Tenure & Household Type & Tenure \\
    \midrule
    (Nation, Detailed): & 3 & 1.92 & 1.92 & 3.84 & 3.84\\
    (State, Detailed): & 3 & 1.92 & 1.92 & 3.84 & 3.84\\
    (County, Detailed): & 11 & 0.14 & 0.14 & 0.28 & 0.28 \\
    (Tract, Detailed):  & 11 & 0.14 & 0.14 & 0.28 & 0.28 \\
    (Place, Detailed):  & 11 & 0.14 & 0.14 & 0.28 & 0.28 \\
    (AIANNH, Detailed): & 11 & 0.14 & 0.14 & 0.28 & 0.28\\
    (Nation, Regional): & 50 & 0.0069 & 0.0069 & 0.0138 & 0.0138\\
    (State, Regional):  & 50 & 0.0069 & 0.0069 & 0.0138 & 0.0138\\
    (County, Regional):  & 50 & 0.0069 & 0.0069 & 0.0138 & 0.0138\\
    (Tract, Regional):  & 50 & 0.0069 & 0.0069 & 0.0138 & 0.0138\\
    (Place, Regional):  & 50 & 0.0069 & 0.0069 & 0.0138 & 0.0138\\
    \bottomrule
    \end{tabular}
   \caption{MOE targets for the statistics released at different population group levels, along with the corresponding privacy loss (unbounded and bounded $\rho$-zCDP for discrete Gaussian). The privacy loss is reported for both household type and tenure. Due to the total population of the United States being published without noise, the bounded privacy-loss budgets in this table were stressed by Census Bureau staff in internal conversations and for presentation to DSEP, for purposes of interpreting the privacy guarantee.}
   \label{tab:moe-targets}
\end{table}

\vspace{\baselineskip}

\noindent \emph{\textbf{Race Multiplicity}}: The stability $\Delta(g_i)$ of the flatmap transformation $g_i$ mapping households to population groups in level $i$ is a significant factor in the noise scale required to satisfy zCDP. The data collection process restricts householders, and hence households, to a maximum of eight detailed race codes and one ethnicity code, which translates to a flatmap stability of nine for any given population group level. This is because an individual with eight unique detailed race codes can be associated with at most eight alone or in any combination characteristic iterations for a level plus the one ethnic characteristic iteration. Higher stability equates to higher variance noise, all else held equal, so an option to improve the noise variance would be to reduce the stability by setting a lower cap on the number of detailed race codes processed for each householder. For example, if householders were restricted to three race codes instead of eight, the stability would drop from nine to four, resulting in a 33\% decrease in MOE when holding the privacy loss constant. However, the restriction would also introduce another form of bias into the statistics and potentially artificially reduce the set of population groups with true positive counts. The DSEP committee opted to leave the race multiplicity parameter at eight.

\vspace{\baselineskip}

\noindent \emph{\textbf{MOE and $\rho$}}:  The SafeTab-H Analysis tool provided an interface for adjusting expected MOE targets to observe the impact on privacy loss, as derived in Section \ref{sec:discrete-gauss-error-bound}. The Census Bureau selected the MOEs and corresponding $\rho$s as displayed in Table \ref{tab:moe-targets} for the production run of SafeTab-H on the 2020 Census data.

\section{Conclusion}

In this paper, we presented the SafeTab-H algorithm, a differentially private algorithm for producing the Detailed DHC-B for the Census Bureau. We covered several key aspects of the algorithm. First, we provided a technical pseudocode description of the algorithm. Then, we covered the privacy properties of the algorithm. Next, we discussed salient differences between the pseudocode and the practical implementation of the algorithm with Tumult Analytics. We covered the key algorithm parameter choices made by Census Bureau policy and the ways in which the SafeTab-H Analysis Tool supported that process.

\bibliographystyle{unsrt}
\bibliography{refs}

\end{document}